\documentclass[12pt, a4paper]{amsart}
\usepackage{amssymb,latexsym,amsmath,amsthm,mathrsfs,yfonts}
\usepackage{fullpage,enumerate}
\input xy
\xyoption{all}
\def\ra{\rightarrow}
\newcommand{\beq}{\begin{eqnarray}}
\newcommand{\eeq}{\end{eqnarray}}

\newcommand\two{\text{I}\!\text{I}}

\def\frakg{\mathfrak g}

\newcommand\RR{\mathbb R}
\newcommand\CC{\mathbb C}
\newcommand\ZZ{\mathbb Z}
\newcommand\TT{\mathbb T}

\newtheorem{theorem}{Theorem}[section]

\theoremstyle{definition}

\theoremstyle{remark}

\begin{document}

\title{T-duality of singular spacetime compactifications\\ in an H-flux}

\author{Andrew Linshaw}
\address[A. Linshaw]{Department of Mathematics,
University of Denver, Denver, CO 80210, USA}
\email{andrew.linshaw@du.edu}

\author{Varghese Mathai}
\address[V. Mathai]{Department of Pure Mathematics, University of Adelaide,
Adelaide, SA 5005, Australia}
\email{mathai.varghese@adelaide.edu.au}

\begin{abstract}
We begin by presenting a symmetric version of the circle equivariant T-duality result in a joint work of the second author with Siye Wu, thereby generalising the results there. We then initiate the study of twisted equivariant Courant algebroids
and equivariant generalised geometry and apply it to our context.
As before, T-duality exchanges type \two A and type \two B string theories.  
In our theory, both spacetime and the T-dual spacetime can be singular spaces when the
fixed point set  is non-empty; the singularities correspond to Kaluza-Klein monopoles.
We propose that the Ramond-Ramond charges of type \two\ string theories on
the singular spaces are classified by twisted equivariant cohomology groups, consistent with the previous work of Mathai and Wu,
and prove that they are naturally isomorphic. We also establish the corresponding isomorphism 
of twisted equivariant Courant algebroids.
\end{abstract}


\thanks{{\em Acknowledgements.}
A.L.\ was supported in part by the grant 318755 from the Simons Foundation.
V.M.\ was supported in part by the Australian Research Council's Discovery
Projects funding scheme  DP150100008.
V.M.\ would like to thank Siye Wu  for relevant
discussions in the past.}

\maketitle
\tableofcontents

\section*{Introduction}
In this paper, we will produce an explicit formula for the topology and
H-flux of the T-dual of possibly {\em singular} compactifications of spacetime in  type \two\, string theory,
significantly generalizing earlier results \cite{BEM, BHM, BS, MW12}. 
Our results apply to T-dualities with respect to {\em any} circle actions.
We start with a smooth manifold $M$ having a commuting pair of circle 
actions $\TT$ and $\hat{\TT}$, where $\TT$ acts on the left and $\hat{\TT}$ acts on the right.
Both spacetime $X=M/\hat\TT$ and the  T-dual spacetime $\hat X= \TT\backslash M$ are 
in general stratified spaces, possibly with
boundaries, in the sense of Goresky-McPherson \cite{GM}.
We also assume that 
we are given a flux $[H] \in H^3_{\hat{\TT}}(M, \ZZ)$ in the equivariant cohomology \cite{AB} of $M$, which represents the cohomology of $X$. Then we will produce an explicit  T-dual 
H-flux $[\widehat H]  \in H^3_{{\TT}}(M, \ZZ)$ on $\hat X$.
The singularities of the spacetime and that of the flux form correspond
to Kaluza-Klein monopoles \cite{Sorkin,Pan}. 
Topological aspects of examples from mirror symmetry in the case of Calabi-Yau manifolds,
cf.~\cite{Hori}, fit into our framework, as well as the earlier results of Mathai-Wu \cite{MW12}.
 Consider the 
commutative diagram,

\begin{equation}\label{correspondencea}
\qquad\xymatrix @=8pc @ur { X = M/ \hat{\TT}  \ar[d]_\pi & 
 M   \ar[d]^{\widehat p} \ar[l]_{p} \\
\TT\backslash M /{\hat\TT} & \widehat  X = {\TT}\backslash M \ar[l]^{\widehat \pi}}
\end{equation}

We can replace the singular spaces by their Borel constructions by lifting the
previous diagram to  
\begin{equation}\label{correspondenceb}
\qquad\xymatrix @=8pc @ur { E\TT \times  M_{\hat\TT}   \ar[d]_\pi & 
E\TT \times M  \times E{\hat\TT} \ar[d]^{\widehat p} \ar[l]_{p} \\
 {}_\TT(M_{\hat\TT}) =  ({}_\TT M)_{\hat\TT}  &  {}_\TT M \times E{\hat\TT} \ar[l]^{\widehat \pi}}.
\end{equation}

In this picture, the flux is $[H] \in H^3(M_{\hat\TT}, \ZZ)= H^3_{\hat\TT}(M, \ZZ)$ and the bundles in the diagram above are principal circle bundles,
so we can apply \cite{BEM} to obtain the unique T-dual flux $[\widehat H] \in H^3({}_\TT M, \ZZ) = H^3_{\TT}(M, \ZZ)$. 

More precisely, the conditions that uniquely specify $[\widehat H]$ are:
\begin{enumerate}
\item $\hat\pi_*([\widehat H])=e_\TT\in H^2_\TT(M_{\hat\TT},\ZZ)$, where $e_\TT=\phi$ is the equivariant  Euler class of the principal circle bundle  
$\pi: E\TT \times  M_{\hat\TT}  \to  {}_\TT(M_{\hat\TT})$.
Also $\pi_*([H])=e_{\hat\TT}\in H^2_{\hat\TT}(M_{\TT},\ZZ)$, where $e_{\hat\TT}=\psi$ is the equivariant  Euler class of the principal circle bundle  
$\hat\pi: {}_\TT M  \times E{\hat\TT}  \to   ({}_\TT M)_{\hat\TT}$.\\

\item $p^*([H])={\hat p}^*([\widehat H])\in H^3(M,\ZZ)$.\\
\end{enumerate}

This is proved using 
the {\em Gysin sequence} of the $\TT$-bundle $\hat\pi\colon E\TT\times M_{\hat\TT}\to  {}_\TT(M_{\hat\TT})$ given by
the long exact sequence
$$ \to H^j(M_{\hat\TT})\stackrel{\hat\pi_*}\to H^{j-1}( {}_\TT(M_{\hat\TT}))\stackrel{\cup\,e_\TT}
   \to H^{j+1}( {}_\TT(M_{\hat\TT}))\stackrel{\hat\pi^*}\to H^{j+1}(M_{\hat\TT})\to                 $$
so the Gysin sequence becomes
$$ \to H^j(M_{\hat\TT})\stackrel{\hat\pi_*}\to H^{j-1}_{\TT}(M_{\hat\TT})\stackrel{\cup\,\phi}
   \to H^{j+1}_{\TT}(M_{\hat\TT})\stackrel{\hat\pi^*}\to H^{j+1}(M_{\hat\TT})\to.                $$
   and the corresponding Gysin sequence for the circle bundle, ${}_\TT M  \times E{\hat\TT} \to  {}_\TT(M_{\hat\TT})$
$$ \to H^j({}_\TT M)\stackrel{\pi_*}\to H^{j-1}_{\hat\TT}({}_\TT M )\stackrel{\cup\,\psi}
   \to H^{j+1}_{\hat\TT}({}_\TT M)\stackrel{\pi^*}\to H^{j+1}({}_\TT M)\to.                $$
   
Choosing a connection $A$ for the circle bundle $\pi$, the curvature $F_A$ is then a basic $2$-form which is equivariantly 
closed, and $[F_A] = e_{\TT}=\phi$. If $H$ is chosen to be an invariant $3$-form on $ E \TT \times M_{\hat{\TT}}  $, 
then as in \cite{BEM}, $H = A \wedge \psi + \Omega$, where $ \Omega$ is a basic $3$-form. Then the T-dual flux 
$\hat H = \hat A \wedge \phi + \Omega$, which an invariant, closed $3$-form on $ {}_\TT M \times E \hat{\TT}$.

Using the definition of twisted equivariant cohomology (cf. \cite{MW12}), we deduce the equivariant T-duality isomorphism
given by the Hori-type formula in \cite{BEM},
\begin{equation}\label{T-duality1}
 H^\bullet({}_\TT M, H) \cong H^{\bullet+1}(M_{\hat\TT}, \widehat H),
\end{equation}
 or equivalently,
 \begin{equation}\label{T-duality2}
 H^\bullet_\TT(M, H) \cong H^{\bullet+1}_{\hat\TT}(M, \widehat H).
\end{equation}
Using the proposal in \cite{MW12} that the Ramond-Ramond charges of the singular spacetimes $X=M/\hat\TT$ and $\hat X= \TT\backslash M$
are classified by the respective twisted equivariant cohomology groups, we then see
 in particular that the Ramond-Ramond charges in these theories are naturally isomorphic.

We remark on the infinite dimensional spaces considered here.  The universal $\TT$-space $E\TT$ 
can be chosen to be the direct limit of odd dimensional spheres $S^\infty= \lim S^{2n+1}$, 
which is a smooth countably compactly generated manifold as studied in section 47.2, \cite{KM}, whose 
tangent space at $p$ is the orthogonal complement of $p\in S^\infty \hookrightarrow \RR^\infty$. 
The Borel construction and other infinite dimensional spaces used here are similarly also either 
smooth countably compactly generated manifolds or obtained from these in a simple way.

Cavalcanti-Gualtieri \cite{CG} showed that the T-duality isomorphism of \cite{BEM},  can be extended to an isomorphism of invariant Courant algebroids in that context. 
In this paper, we further extend this to the case considered in this paper, as discussed above when the circle actions are not necessarily free. We show in section \ref{sect:equiv exact courant} that 
T-duality as described above, extends to an isomorphism of  
equivariantly exact Courant algebroids. Then in section \ref{sect:weil model}, we introduce the Weil model  for equivariantly exact Courant algebroids, 
which is a small model for it, and show in section \ref{sect:weil T-duality} that T-duality extends to an isomorphism of  the Weil models for equivariantly 
exact Courant algebroids. 

Recall that the global aspects of T-duality in type \two\,
string theories, involve compactifications of spacetime $X$ with an H-flux. 
The local transformation rules of the low energy effective fields under
T-duality, are known as the Buscher rules \cite{Bus}. 
However, in cases in which there is a topologically nontrivial NS 3-form
H-flux, the Buscher rules only make sense locally.
Examples of T-duals to such backgrounds were studied in 
\cite{AABL,DLP,DRS,GLMW,KSTT}, but the general formula for the topology and
H-flux of the T-dual with respect to any free circle action was presented for the first time in \cite{BEM}, generalizing \cite{Hori1}.  
As supporting evidence, it was shown there that locally, the said formula
agrees with the Buscher rules and that globally it yields an isomorphism
of the twisted cohomology and twisted K-theory-valued conserved Ramond-Ramond
charges \cite{MM,Wib,BM,BS,Pan}. Recent followup work on T-duality include \cite{HM15,LM15}.

\section{A concrete example}
We begin by giving a concrete example of a manifold $M$ with commuting actions of $\TT$ and $\hat{\TT}$, such that the T-dual spaces $X=M/\hat\TT$ and $\hat X= \TT\backslash M$ are both singular spaces with nontrivial fluxes. This shows that the setup in our paper is indeed more general than that of \cite{MW12}.


Let $M = S^3$, regarded as the joint of two circles $S^1 \star S^1$. Consider the 
circle actions $\TT$ by rotation on the first factor, and $\hat{\TT}$ by rotation on the second factor, and assume that they rotate by different speeds. Then the spacetime $X = S^3/ {\hat \TT}$ will be singular, since the $\hat{\TT}$-fixed point set in $M$ is non-empty. Similarly, the T-dual spacetime $\hat{X} = \TT\backslash S^3$ is also singular. Let $H$ denote $k$ times the invariant volume form on $S^3$. 

Next, we show that $S^3$ with the $\TT$-action is equivariantly formal. Recall that the cohomology ring $H^*(B\TT) = H^*(\CC P^{\infty})= \ZZ [t]$ where degree$(t)=2$. Let $D_+$ respectively $D_-$ denote the open upper and lower hemispheres determined by the axis of the $\TT$-action. Both $D_+$ and $D_-$ are $\TT$-invariant, contractible sets, and $H^*_{\TT}(D_{\pm})= H^*_{\TT}(pt) = H^*(B \TT)  = \ZZ[t]$. Also, $\TT$ acts freely on the overlap $D_+ \cap D_-$, so 
$$H_{\TT}^*(D_+ \cap D_-) = H^*((D_+\cap D_-)/\TT)= H^*(pt).$$ We can apply the Mayer-Vietoris exact sequence
to the decomposition 
$$S^3\times_\TT E\TT = (D_+\times_\TT E\TT) \cup (D_-\times_\TT E\TT)$$
and use the previous observations to obtain the long exact sequence,
$$
\ldots\to H^{*-1}(pt) \to H^*_\TT(S^3) \to H^*(B\TT) \oplus H^*(B\TT)  \to H^*(pt) \to \ldots
$$
Comparing dimensions, we conclude that $S^3$ is equivariantly formal for the given action of $\TT$.
An analogous argument shows that $S^3$ is also equivariantly formal for $\hat\TT$-action. In particular
$$H^3_{\TT}(S^3;\ZZ) \cong H^3(S^3;\ZZ) \cong H^3_{\hat\TT}(S^3;\ZZ),$$
and so $[H] \in H^3(S^3;\ZZ)$ determines both $[H_1] \in H^3_{\TT}(S^3;\ZZ)$ and $[H_2] \in H^3_{\hat\TT}(S^3;\ZZ)$, which are nontrivial. We have the following commutative diagram representing T-duality, where the singular spacetimes $S^3/{\hat\TT}$ and $\TT \backslash S^3$ are T-dual.
\begin{equation}
\qquad\xymatrix @=8pc @ur { X = S^3/{\hat\TT}  \ar[d]_\pi & 
 S^3   \ar[d]^{\widehat p} \ar[l]_{p} \\
\TT\backslash   S^3 /{\hat\TT} & \widehat  X =  \TT\backslash S^3  \ar[l]^{\widehat \pi}}
\end{equation}







\section{Hori type formulae} \label{sect:hori}

Here we use the Weil model of twisted equivariant cohomology and write down the Hori type formulae inducing
the T-duality isomorphisms \eqref{T-duality1} and \eqref{T-duality2} on the level of representatives.

We first set up notation. Let $\frakg$ denote the Lie algebra of $\TT$. Consider the Weil algebras $$W= W(\frakg) = S(\frakg^*) \otimes \Lambda(\frakg^*),\qquad 
\hat W= W(\hat\frakg) = S(\hat\frakg^*) \otimes \Lambda(\hat\frakg^*).$$ Here $W$ is the differential graded algebra (DGA) generated by a degree one
element $\theta$ and a degree 2 element $\mu$ satisfying the relations,
\begin{align*}
& \iota_W \theta=1, \qquad   \iota_W \mu =0,\\
& d_W \theta = \mu, \qquad d_W \mu=0.
\end{align*}
Recall that the {\it Weil model} for $H^*_{\TT}(M)$ is the cohomology $$H^*((\Omega(M) \otimes W)_{\text{bas}}, d).$$ Here $d$ is the total differential $d_M + d_W$ where $d_M$ is the de Rham differential, and the basic subcomplex $(\Omega(M) \otimes W)_{\text{bas}}$ is the subspace which is $\TT$-invariant and annihilated by the operator $\iota_V + \iota_{W}$. Here $V$ is the vector field on $M$ generated by the action of $\TT$. Then $H^*((\Omega(M) \otimes W)_{\text{bas}}, d) \cong H^*_{\TT}(M)$ with real coefficients. Note that by identifying $\Omega(M) \otimes W$ with a subcomplex of $\Omega(M \times E\TT) \cong \Omega(M) \otimes \Omega(E\TT)$, we can regard $(\Omega(M) \otimes W)_{\text{bas}}$ as a subcomplex of $\Omega(M \times E\TT)_{\text{bas}} \cong \Omega(M_{\TT})$.

Similarly, let $\hat\frakg$ denote the Lie algebra of $\hat\TT$, and let $\hat W $ be the DGA generated by a degree one
element $\hat\theta$ and a degree 2 element $\hat\mu$ satisfying the relations,
\begin{align*}
& \iota_{\hat{W}} \hat\theta=1, \qquad   \iota_{\hat{W}} \hat\mu =0,\\
& d_{\hat{W}}\hat\theta = \hat\mu, \qquad d_{\hat{W}} \hat\mu=0.
\end{align*}
We shall use the notation $\TT-\text{bas}$ and $\hat{\TT}-\text{bas}$ to denote the basic subcomplexes with respect to $\TT$ and $\hat{\TT}$, respectively. In diagram \ref{correspondenceb}, the following are the Weil models for the associated spaces.
\begin{enumerate}
\item $\Omega(M) \otimes W \otimes \hat W$ is the Weil model for $E\TT \times M  \times E{\hat\TT} $,
\item $(\Omega(M) \otimes \hat W)_{ \hat{\TT}-\text{bas}} \otimes W$  is the Weil model for $M_{\hat\TT}  \times E{\TT} $,
\item $(\Omega(M) \otimes W)_{\TT-\text{bas}} \otimes \hat W$  is the Weil model for $ {}_\TT M  \times E{\hat\TT} $,
\item $(\Omega(M) \otimes W \otimes \hat W)_{\TT\times \hat{\TT}-\text{bas}}$  is the Weil model for $ {}_\TT M_{\hat\TT} $.
\end{enumerate}
Note that if we replace the Weil models (1), (2), and (3) above by their invariant subspaces with respect to $\TT\times \hat{\TT}$, $\TT$, and $\hat{\TT}$, respectively, the cohomology is unchanged. The corresponding commutative diagram of Weil models is then

\begin{equation}
\qquad\xymatrix @=8pc @ur { ((\Omega(M) \otimes \hat W)_{\hat{\TT}-\text{bas}} \otimes W)^{\TT}  \ar[d]_\iota & 
(\Omega(M) \otimes W \otimes \hat W)^{\TT \times \hat{\TT}} \ar[d]^{\iota} \ar[l]_{\hat\iota} \\
 (\Omega(M) \otimes W \otimes \hat W)_{\TT\times \hat{\TT}-\text{bas}}  &  ((\Omega(M) \otimes W)_{\TT-\text{bas}} \otimes \hat W)^{\hat{\TT}} \ar[l]^{\hat\iota}}.
\end{equation}
Here $\iota = \iota_V + \iota_W$ and $\hat\iota = \iota_{\hat{V}} + \iota_{\hat{W}}$, where $V$ and $\hat{V}$ are the vector field on $M$ generated by the action of $\TT$ and $\hat{\TT}$, respectively. We now have total differential $d = d_M + d_W + d_{\hat{W}}$. Suppose that $H \in (\Omega(M) \otimes W \otimes \hat W)^{\TT\times \hat{\TT}}$ satisfies $d H=0$ and $\hat{\iota} H = 0$, so that
$$
H \in ((\Omega(M) \otimes \hat W)_{\hat{\TT}-\text{bas}} \otimes W)^{\TT}.
$$

Let $A$ be a connection on the principal $\TT$-bundle $\pi$ as in diagram \ref{correspondenceb}. Then 
$\iota A =1$ and $(d\iota + \iota d)A=0$ so that $\iota d A=0$. Therefore   
$$
A \in ((\Omega(M) \otimes \hat W)_{\hat{\TT}-\text{bas}} \otimes W)^{\TT}.
$$
The curvature of $A$ is $F = dA$, where $dF=0$ and $\iota F=0$ so that $(d\iota + \iota d)F=0$. Therefore
$$
F \in (\Omega(M)  \otimes W \otimes \hat W)_{\TT\times \hat{\TT} - \text{bas}}.
$$
We can write the connection $A$ in terms of the universal connection in the form 
$$A = f \theta + \alpha,$$ where $f$ is a smooth function on $M$, $\theta$ is the connection on the circle bundle  $E\TT \to B\TT,\, \alpha$ is a $1$-form on $M$. Then the generator $V$ of the circle action on $E\TT \times M$ is of the form $V = V_1 + V_2$, where $V_1$ is the generator of the circle action on $E\TT$ and $V_2$ is the generator of the circle action on $M$. Since $A$ is a connection, we must have 
$$
0 = L_V A \Rightarrow V_2(f)=0, \text{and}\, L_{V_2} \alpha=0.
$$
Also $\iota_V(A)=1 \Rightarrow \iota_{V_2} \alpha = 1- f.$

Similarly let $\hat A$ be a connection on the principal $\hat\TT$-bundle $\hat\pi$ as in diagram \ref{correspondenceb}. Then 
$\hat\iota \hat A =1$ and $(d\hat\iota + \hat\iota d)\hat A=0$ so that $\hat\iota d \hat A=0$. Therefore   
$$
\hat A \in ((\Omega(M) \otimes W)_{\TT - \text{bas}} \otimes \hat W)^{\hat{\TT}}.
$$
Similarly, if $\hat F= d \hat A$ is the curvature of $\hat A$, then $d\hat F=0$ and $\hat\iota \hat F=0$ so that $(d\hat\iota + \hat\iota d)\hat F=0$. Therefore
$$
\hat F \in (\Omega(M)  \otimes W \otimes \hat W)_{\TT\times \hat{\TT} - \text{bas}}.
$$
The one constraint imposed on $\hat A$ is that $\iota H = \hat F$. Then $H = A \wedge \hat F + \Omega$ for some 
$$
\Omega \in (\Omega(M)  \otimes W \otimes \hat W)_{\TT\times \hat{\TT} - \text{bas}}.
$$
We now define the T-dual flux by $$\hat H = F \wedge \hat A + \Omega.$$ Then $\hat\iota \hat H= F$, $\iota H = 0$, $d\hat H=0$, and 
$(d\hat\iota + \hat\iota d)\hat H=0$ so that 
$$
\hat H \in ((\Omega(M) \otimes W)_{\TT - \text{bas}} \otimes \hat W)^{\hat{\TT}}.
$$

Next we describe the Hori formula in \cite{BEM} in terms of the Weil complex.
Let $$G \in ((\Omega(M) \otimes \hat W)_{\hat{\TT} - \text{bas}} \otimes W)^{\TT}.$$ Then the Hori formula in this context is given by the formula,
$$
T(G) = \iota( ( 1 + A \wedge \hat A)G).
$$
Then one checks that
$T(G) \in  ((\Omega(M) \otimes W)_{\TT - \text{bas}} \otimes \hat W)^{\hat{\TT}}$, and that the map 
\begin{equation} \label{tdualitylinear} T: ((\Omega(M) \otimes \hat W)_{\hat{\TT} - \text{bas}} \otimes W)^{\TT} \ra  ((\Omega(M) \otimes W)_{\TT - \text{bas}} \otimes \hat W)^{\hat{\TT}} \end{equation} is a linear isomorphism, cf. \cite{MW11}. Moreover if $d_H = d + H\wedge$ is the twisted differential, and 
$d_H G=0$, then $d_{\hat H} T(G)=0$, so that $T$ is a chain map, defining
a map in twisted equivariant cohomology,
$$
T : H_{\hat\TT}^\bullet(M, H) \longrightarrow H_{\TT}^{\bullet-1}(M, \hat H) 
$$
which was shown in \cite{BEM} to be an isomorphism.

\section{Equivariant standard and exact Courant algebroids}\label{sect:equiv exact courant}

We will consider two equivalent models of equivariant Courant algebroids and their twisted analogues. First, given a manifold $M$, the standard Courant algebroid is defined to be $E= TM \oplus T^*M$, with Courant bracket 
$$\langle X + \xi, Y + \eta\rangle = \frac{1}{2}( \iota_X \eta + \iota_Y \xi),$$ $$[X + \xi, Y + \eta] = [X,Y] + L_X\eta - L_Y \xi - \frac{1}{2} d(\iota_X \eta - \iota_Y \xi).$$
$E$ is also known as the {\em generalized tangent bundle} in generalized geometry. 

More generally, one considers {\em exact} equivariant Courant algebroids $E$, that is, there is an exact sequence of vector bundles
$$0\ra T^*M \ra E \ra TM \ra 0,$$ where the last map is $\pi$. Exact Courant algebroids are classified by $H^3(M)$, \cite{Severa1,Severa2}. For every such $E$, there is a splitting $E \cong TM \oplus T^*M$ and a closed $3$-form $H \in \Omega^3(M)$ such that the bilinear form and bracket are given by
\begin{equation} \label{twistedbracket} [X + \xi, Y + \eta]_H = [X,Y] + L_X\eta - L_Y \xi - \frac{1}{2} d(\iota_X \eta - \iota_Y \xi)+ \iota_Y \iota_X H.\end{equation}
For any $H$, $\Omega^*(M)$ is a Clifford module over $TM \oplus T^*M$ via $$(X+\xi)\cdot \omega = \iota_X (\omega) + \xi \wedge \omega.$$ Define the {\it twisted de Rham differential} $d_H$ on $\Omega^*(M)$ by $$d_H(\omega) = d(\omega) + H\wedge \omega.$$ Note that $d_H$ is not a derivation of $\Omega^*(M)$ regarded as an algebra, but it is a derivation of $\Omega^*(M)$ regarded as a left $\Omega^*(M)$-module. In particular, for homogeneous $a,\omega \in \Omega^*(M)$, we have $d_H(a\omega) = d(a) \omega + (-1)^{|a| |\omega|} a d_H(\omega)$. Finally, for $X,Y \in \text{Vect}(M)$, $\xi,\eta \in \Omega^1(M)$, and $\omega \in \Omega^*(M)$, we have
\begin{equation} \label{twistedbracket} [X+\xi, Y+\eta]_H \cdot \omega = [[d_H,X+\xi],Y+\eta]\cdot \omega. \end{equation}

Suppose now that $M$ has commuting actions of $\TT$ and $\hat{\TT}$, and we have T-dual fluxes $[H] \in H^3(M_{\hat{\TT}},\mathbb{Z})$ and $[\hat{H}] \in H^3({}_\TT M,\mathbb{Z})$. Also, assume that we have chosen connections $A$ and $\hat{A}$ on the bundles $\pi$ and $\hat{\pi}$ in diagram \eqref{correspondenceb}, such that $\iota H = \hat{F} = d \hat{A}$ and $\hat{\iota} \hat{H} = F = dA$. Then the quotients $$\big(T(E\TT \times M_{\hat{\TT}}) \oplus T^*(E\TT \times M_{\hat{\TT}})\big)/\TT,\qquad \big(T({}_\TT M  \times E \hat{\TT}) \oplus T^*({}_\TT M \times E \hat{\TT})\big)/\hat{\TT}$$ 
are again (twisted) Courant algebroids. By a result of Cavalcanti-Gualtieri \cite{CG}, there is an isomorphism of twisted Courant algebroids
\begin{equation} \label{cavgual} \tau: \big(T(E\TT \times M_{\hat{\TT}}) \oplus T^*(E\TT \times M_{\hat{\TT}})\big)/\TT \ra \big(T({}_\TT M  \times E \hat{\TT}) \oplus T^*({}_\TT M  \times E \hat{\TT})\big)/\hat{\TT},\end{equation}
that is compatible with the previously described T-duality isomorphism \eqref{tdualitylinear}. This map can be described explicitly as follows. The connection form $A$ on $E\TT \times M_{\hat{\TT}}$ gives rise to a splitting of the space of $\TT$-invariant $1$-forms $\Omega^1(E\TT \times M_{\hat{\TT}})^{\TT}$. Such a $1$-form can be uniquely written in the form $\omega = \omega_0 + fA$ where $\omega_0 \in \Omega^1({}_{\TT} M_{\hat{\TT}})$ and $f \in C^{\infty}({}_{\TT} M_{\hat{\TT}})$. Similarly, if we fix a $\TT$-invariant vector field $Y_A$ on $E\TT \times M_{\hat{\TT}}$ which is dual to $A$ in the sense that $\iota_{Y_A}(A) = 1$, we can write a $\TT$-invariant vector field on $E\TT \times M_{\hat{\TT}}$ in the form $Y = Y_0 + g Y_A$, where $g \in C^{\infty}(E\TT \times M_{\hat{\TT}})$, and $Y_0$ satisfies $\iota_{Y_0}(A) = 0$. Then the isomorphism \eqref{cavgual} is given by
 \begin{equation} \label{defoftau} \tau (Y_0 + g Y_A, \omega_0 + f A) = (Y_0 + f Y_{\widehat{A}}, \omega_0 + g \widehat{A}).\end{equation}

\section{Weil model for equivariant Courant algebroids}\label{sect:weil model}

Based on the Weil model of equivariant cohomology, we give a simpler notion of a (twisted) equivariant Courant algebroid on $M$. As in Section \ref{sect:hori}, suppose that $M$ has commuting actions of $\TT$ and $\hat{\TT}$ and we have T-dual fluxes $[H] \in H^3(M_{\hat{\TT}},\mathbb{Z})$ and $[\hat{H}] \in H^3(M_{\TT},\mathbb{Z})$. Also, assume that we have chosen connections $A$ and $\hat{A}$ on the bundles $\pi$ and $\hat{\pi}$ in diagram \eqref{correspondenceb}, such that $\iota H = \hat{F} = d \hat{A}$ and $\hat{\iota} \hat{H} = F = dA$. 
 
First, we define the Weil model of $T(E\TT \times M \times E \hat{\TT}) \oplus T^*(E\TT \times M \times E \hat{\TT})$, which will be a Courant algebroid on $M$. Define a vector bundle
$$ T^*M \oplus \mathbb{I}_{\theta} \oplus \mathbb{I}_{\hat{\theta}}$$ on $M$, where $\mathbb{I}_{\theta}$ and $\mathbb{I}_{\hat{\theta}}$ are trivial rank one vector bundles on $M$. The space of smooth sections $\Gamma(T^*M \oplus \mathbb{I}_{\theta} \oplus \mathbb{I}_{\hat{\theta}})$ is just $\Omega^1(M) \oplus C^{\infty}(M)\theta \oplus C^{\infty}(M) \hat{\theta}$, where $\theta$ and $\hat{\theta}$ are constant sections of $\mathbb{I}_{\theta}$ and $\mathbb{I}_{\hat{\theta}}$ respectively.

We can identify $\Gamma(T^*M \oplus \mathbb{I}_{\theta}\oplus \mathbb{I}_{\hat{\theta}})$ with a subspace of $\Gamma(T^*(E\TT \times M \times E \hat{\TT}))$ as follows. Let $$\pi_1: E\TT \times M \times E \hat{\TT} \ra E\TT,\qquad \pi_2: E\TT \times M \times E \hat{\TT} \ra M,\qquad \pi_3:E\TT \times M \times E \hat{\TT} \ra E \hat{\TT},$$ denote the projection maps. Then $\Gamma(T^*M \oplus \mathbb{I}_{\theta} \oplus \mathbb{I}_{\hat{\theta}})$ can be identified with 
$$\pi^*_2(\Gamma(T^*M)) \bigoplus \bigg(\pi^*_2(C^{\infty}(M)) \otimes \pi^*_1(\Gamma_0(\mathbb{I}))\bigg) \bigoplus \bigg(\pi^*_2(C^{\infty}(M)) \otimes \pi^*_3(\Gamma_0(\hat{\mathbb{I}}))\bigg).$$ Here $\mathbb{I}$ and $\hat{\mathbb{I}}$ are the trivial rank one subbundles generated by the connection form $\theta$ on $E\TT$, and the connection form $\hat{\theta}$ on $E \hat{\TT}$, respectively, and $\Gamma_0(\mathbb{I})$ and $\Gamma_0(\hat{\mathbb{I}})$ denote the spaces of constant sections. Let $\Theta = \pi_1^*(\theta)$ and $\hat{\Theta} = \pi_3^*(\hat{\theta})$ be the corresponding forms in $\Omega^1(E\TT \times M \times E \hat{\TT})$.

Next, we define another vector bundle 
$$TM \oplus \mathbb{I}_{\theta^*} \oplus \mathbb{I}_{\hat{\theta}^*}$$ on $M$, where $\mathbb{I}_{\theta^*}$ and $\mathbb{I}_{\hat{\theta}^*}$ denote again trivial rank one bundles on $M$. We can identify $\Gamma(TM \oplus \mathbb{I}_{\theta^*} \oplus \mathbb{I}_{\hat{\theta}^*})$ with a subspace of $\Gamma(T(E\TT \times M \times E \hat{\TT}))$ as follows. First, a section of $TM \oplus \mathbb{I}_{\theta^*} \oplus \mathbb{I}_{\hat{\theta}^*}$ has the form $$Y + f \theta^* + g \hat{\theta}^*,$$ where $Y$ is a smooth vector field on $M$, $f, g \in C^{\infty}(M)$, and $\theta^*$, $\hat{\theta}^*$ are constant sections of $\mathbb{I}_{\theta^*}$, $\mathbb{I}_{\theta^*}$, respectively. 

Choose vector fields $\Theta^*$ and $\hat{\Theta}^*$ on $E\TT \times M \times E \hat{\TT}$ which are dual to $\Theta$ and $\hat{\Theta}$, respectively. We may identify $Y + f \theta^* + g \hat{\theta}^*$ with a section of $T(E\TT \times M \times E \hat{\TT})$ of the from 
$$\tilde{Y} + \pi^*_2(f) \Theta^* +  \pi^*_2(g) \hat{\Theta}^*,$$ where $\tilde{Y}$ is a vector field on $E\TT \times M \times E \hat{\TT}$ satisfying $(\pi_2)_*(\tilde{Y}) = Y$. Furthermore, we may choose $\tilde{Y}$ such that 
\begin{equation} \label{eq:condition1} [\tilde{Y}, \Theta^*] = 0= [\tilde{Y}, \hat{\Theta}^*]  ,\qquad \iota_{\tilde{Y}}(\Theta) = 0 =\iota_{\tilde{Y}}(\hat{\Theta}).\end{equation} 

Note that if $\tilde{Y}_1, \tilde{Y}_2$ satisfy \eqref{eq:condition1}, so does $[\tilde{Y}_1, \tilde{Y}_2]$. Also, if $f \in C^{\infty}(M)$, then for all $f \in C^{\infty}(M)$,
\begin{equation} \label{eq:condition2}\tilde{Y} (\pi_2^*(f)) = \pi_2^*(Yf).\end{equation}
Given two sections $Y + f_1 \theta^* + f_2 \hat{\theta}^*$ and $Y' + f'_1 \theta^*+ f'_2 \hat{\theta}^*$ of $TM \oplus \mathbb{I}_{\theta^*}\oplus \mathbb{I}_{\hat{\theta}^*}$, we have the bracket operation $$[Y + f_1 \theta^* + f_2 \hat{\theta}^*, Y' + f'_1 \theta^* + f'_2 \hat{\theta}^*] = [Y,Y']+ (Yf'_1 - Y' f_1)\theta^* +  (Yf'_2- Y' f_2) \hat{\theta}^*,$$ which is compatible with the bracket in $T(E\TT \times M \times E \hat{\TT})$ under the above identification by \eqref{eq:condition1} and \eqref{eq:condition2}.

Next, we define a Courant bracket $[,]$ on $(TM \oplus \mathbb{I}_{\theta^*} \oplus \mathbb{I}_{\hat{\theta}^*}) \oplus (T^*M \oplus \mathbb{I}_{\theta} \oplus \mathbb{I}_{\hat{\theta}})$ by
\begin{equation} \label{eq:courant} \begin{split} [Y + f_1 \theta^* + f_2 \hat{\theta}^*+ \xi + g_1 \theta + g_2 \hat{\theta}, Y' + f_1' \theta^* + f'_2 \hat{\theta}^*+ \xi' + g'_1 \theta + g'_2 \hat{\theta}]
\\ =  [Y,Y'] +  (Y f'_1- Y' f_1)\theta^* +  (Y f'_2- Y' f_2)\hat{\theta}^* 
\\ + L_Y \xi' - L_{Y'} \xi + d \bigg( \iota_Y \xi' - \iota_{Y'} \xi + f_1g'_1 - f'_1g_1 + f_2g'_2 - f'_2g_2 \bigg)
\\ + (Yg'_1 - Y' g_1) \theta + (Yg'_2 - Y' g_2) \hat{\theta} . \end{split} \end{equation} Here $f_1, f_2,g_1, g_2,f'_1, f'_2, g'_1, g'_2$ are smooth functions on $M$, $Y,Y'$ are smooth vector fields on $M$, and $\xi, \xi'$ are smooth $1$-forms on $M$. To see that this defines a Courant algebroid structure, it suffices to note that if we identify $\Gamma(TM \oplus \mathbb{I}_{\theta^*}\oplus \mathbb{I}_{\hat{\theta^*}}) \oplus \Gamma(T^*M \oplus \mathbb{I}_{\theta}\oplus \mathbb{I}_{\hat{\theta}})$ with a subspace of $\Gamma(T(E\TT \times M \times E \hat{\TT}) )\oplus \Gamma(T^*(E\TT \times M \times E \hat{\TT}))$ as above, this subspace is closed under the Courant bracket, which is given by \eqref{eq:courant}.

\section{Basic and invariant subbundles}

First, we define $(T^*M \oplus \mathbb{I}_{\theta} \oplus \mathbb{I}_{\hat{\theta}})_{\TT-\text{bas}}$ to be the subbundle of $T^*M \oplus \mathbb{I}_{\theta}\oplus \mathbb{I}_{\hat{\theta}}$ whose space of smooth sections is $\TT$-invariant and annihilated by $\iota_V+ \iota_{\theta}$. In this notation, $\iota_V$ is the contraction on $T^*M$ with respect to the vector field $V$ generated by the $\TT$-action, and $\iota_{\theta}$ is defined by $\iota_{\theta} (\xi + f \theta + g \hat{\theta}) = f$.

Under the identification of $\Gamma(T^*M \oplus \mathbb{I}_{\theta}\oplus \mathbb{I}_{\hat{\theta}})$ with the degree one part of $\Omega(M) \otimes W \otimes \hat{W}$, this identifies $\Gamma((T^*M \oplus \mathbb{I}_{\theta}\oplus \mathbb{I}_{\hat{\theta}})_{\TT-\text{bas}})$ with the degree one part of $(\Omega(M) \otimes W)_{\TT-\text{bas}} \otimes \hat{W}$. This allows us to identify $\Gamma((T^*M \oplus \mathbb{I}_{\theta} \oplus \mathbb{I}_{\hat{\theta}})_{\TT-\text{bas}})$ with a subspace of $\Gamma(T^*({}_{\TT} M \times E\hat{\TT}))$.

Similarly, we define $(TM \oplus \mathbb{I}_{\theta^*} \oplus \mathbb{I}_{\hat{\theta}^*})_{\text{bas}}$ to be the subbundle of $(TM \oplus \mathbb{I}_{\theta^*}\oplus \mathbb{I}_{\hat{\theta}^*})$ whose space of smooth sections $Y + f \theta^*+ g \hat{\theta}^*$ is $\TT$-invariant and satisfies
$$\iota_{Y + f \theta^*+ g \hat{\theta}^*}(\theta) = 0.$$ Under the identification of $\Gamma(TM \oplus \mathbb{I}_{\theta^*} \oplus \mathbb{I}_{\hat{\theta}^*})$ with a subspace of $\Gamma (T(E\TT \times M \times E \hat{\TT}))$, this identifies $\Gamma((TM \oplus \mathbb{I}_{\theta^*} \oplus \mathbb{I}_{\hat{\theta}^*})_{\TT-\text{bas}})$ with a subspace of $\Gamma(T({}_{\TT} M \times E\hat{\TT}))$.

Finally, we define
 \begin{equation} \label{eq:courantbas} \big((TM \oplus \mathbb{I}_{\theta^*}\oplus \mathbb{I}_{\hat{\theta}^*})  \oplus (T^*M \oplus \mathbb{I}_{\theta} \oplus \mathbb{I}_{\hat{\theta}})\big)_{\TT-\text{bas}} = (TM \oplus \mathbb{I}_{\theta^*}\oplus \mathbb{I}_{\hat{\theta}^*})_{\TT-\text{bas}} \oplus (T^*M \oplus \mathbb{I}_{\theta}\oplus \mathbb{I}_{\hat{\theta}})_{\TT-\text{bas}}.\end{equation} Then the $\hat{\TT}$-invariant subspace 
\begin{equation} \label{eq:courantbasinv} \big((TM \oplus \mathbb{I}_{\theta^*}\oplus \mathbb{I}_{\hat{\theta}^*})  \oplus (T^*M \oplus \mathbb{I}_{\theta} \oplus \mathbb{I}_{\hat{\theta}})\big)^{\hat{\TT}}_{\TT-\text{bas}}\end{equation} is closed under the Courant bracket \eqref{eq:courant}, and can be identified with a Courant subalgebroid of $$\big(T({}_{\TT} M \times E \hat{\TT}) \oplus T^*({}_{\TT} M \times E \hat{\TT})\big)/\hat{\TT},$$ as above. Note that the Courant algebroid \eqref{eq:courantbasinv} acts by Clifford multiplication on the $\mathbb{Z}/2\mathbb{Z}$-graded complex 
$$\big((\Omega(M) \otimes W)_{\TT-\text{bas}} \otimes \hat{W} \big)^{\hat{\TT}}.$$
There is a similar definition of the Courant algebroid
 $$\big((TM \oplus \mathbb{I}_{\theta^*}\oplus \mathbb{I}_{\hat{\theta}^*})  \oplus (T^*M \oplus \mathbb{I}_{\theta} \oplus \mathbb{I}_{\hat{\theta}})\big)^{\TT}_{\hat{\TT}-\text{bas}},$$ which acts by Clifford multiplication on the $\mathbb{Z}/ 2\mathbb{Z}$-graded complex
 $$\big((\Omega(M) \otimes \hat{W})_{\hat{\TT}-\text{bas}} \otimes W  \big)^{\TT}.$$

\section{Twisted versions}

Let $[H] \in H^3(M_{\hat{\TT}},\mathbb{Z})$ and $[\hat{H}] \in H^3(M_{\TT},\mathbb{Z})$ be T-dual fluxes as above. Fix representatives $H \in ((\Omega(M) \otimes \hat{W})_{\hat{\TT}-\text{bas}} \otimes W)^{\TT}$ and $\hat{H} \in ((\Omega(M) \otimes W)_{\TT-\text{bas}} \otimes  \hat{W})^{\hat{\TT}}$, which are equivariantly closed $3$-forms in the Weil model.
 
Recall that as a $\mathbb{Z}/2\mathbb{Z}$-graded complex, $((\Omega(M) \otimes \hat{W})_{\hat{\TT}-\text{bas}} \otimes W)^{\TT}$ has the $H$-twisted differential $$d_H = d+ H \wedge  = d_M + d_W + d_{\hat{W}}  + H\wedge,$$ and the $\mathbb{Z}/2\mathbb{Z}$-graded cohomology $$H^{\bullet}(((\Omega(M) \otimes \hat{W})_{\hat{\TT}-\text{bas}} \otimes W)^{\TT}), d_H)$$ is called the {\it twisted $\TT$-equivariant cohomology} of $M$. It coincides with the usual twisted cohomology $H^{\bullet}(M_{\TT},H)$. Similarly, $((\Omega(M) \otimes W)_{\TT-\text{bas}} \otimes  \hat{W})^{\hat{\TT}}$ has the $\hat{H}$-twisted differential $d_{\hat{H}} = d_M + d_W + d_{\hat{W}} +  \hat{H} \wedge$, and its $\mathbb{Z}/2\mathbb{Z}$-graded cohomology is just $H^{\bullet}(M_{\hat{\TT}},\hat{H})$.

There is also an $H$-twisted version $[,]_H$ of the Courant algebroid \eqref{eq:courantbasinv}. It has bracket
\begin{equation} \label{eq:twistedcourant} \begin{split} [Y + f_1 \theta^* + f_2 \hat{\theta}^*+ \xi + g_1 \theta + g_2 \hat{\theta}, Y' + f_1' \theta^* f'_2 \hat{\theta}^*+ \xi' + g'_1 \theta + g'_2 \hat{\theta}]_H & \\ =  [Y + f_1 \theta^* + f_2 \hat{\theta}^*+ \xi + g_1 \theta + g_2 \hat{\theta}, Y' + f_1' \theta^* f'_2 \hat{\theta}^*+ \xi' + g'_1 \theta + g'_2 \hat{\theta}]& \\   + \iota_{Y'+ f'_1\theta^* + f'_2 \hat{\theta}^*} \iota_{Y+f_1 \theta^*+f_2 \hat{\theta}^*} H.\end{split} \end{equation}
As in the untwisted case, the $H$-twisted version of \eqref{eq:courantbasinv} acts by Clifford multiplication on $$\big((\Omega(M) \otimes W)_{\TT-\text{bas}} \otimes \hat{W} \big)^{\hat{\TT}}.$$ The action is compatible with the twisted differential $d_H$, in the sense of \eqref{twistedbracket}. Finally, there is a similar $\hat{H}$-twisted version of $$\big((TM \oplus \mathbb{I}_{\theta^*}\oplus \mathbb{I}_{\hat{\theta}^*})  \oplus (T^*M \oplus \mathbb{I}_{\theta} \oplus \mathbb{I}_{\hat{\theta}})\big)^{\TT}_{\hat{\TT}-\text{bas}}.$$

\section{T-duality of Weil models}\label{sect:weil T-duality}

Suppose that $M$ has commuting actions of $\TT$ and $\hat{\TT}$, and we have T-dual fluxes $[H] \in H^3(M_{\hat{\TT}},\mathbb{Z})$ and $[\hat{H}] \in H^3(M_{\TT},\mathbb{Z})$ as in Section \ref{sect:hori}. Also, assume that we have chosen connections $A$ and $\hat{A}$ on the bundles $\pi$ and $\hat{\pi}$ in diagram \eqref{correspondenceb}, such that $\iota H = \hat{F} = d \hat{A}$ and $\hat{\iota} \hat{H} = F = dA$.

\begin{theorem} 
Suppose that $M$ has commuting actions of $\TT$ and $\hat{\TT}$ and we have T-dual fluxes $[H] \in H^3(M_{\hat{\TT}},\mathbb{Z})$ and $[\hat{H}] \in H^3(M_{\TT},\mathbb{Z})$. Then we have an isomorphism of twisted Courant algebroids
$$\big((TM \oplus \mathbb{I}_{\theta^*}\oplus \mathbb{I}_{\hat{\theta}^*})  \oplus (T^*M \oplus \mathbb{I}_{\theta} \oplus \mathbb{I}_{\hat{\theta}})\big)^{\TT}_{\hat{\TT}-\text{bas}} \ra \big((TM \oplus \mathbb{I}_{\theta^*}\oplus \mathbb{I}_{\hat{\theta}^*})  \oplus (T^*M \oplus \mathbb{I}_{\theta} \oplus \mathbb{I}_{\hat{\theta}})\big)^{\hat{\TT}}_{\TT-\text{bas}}.
$$
which is compatible with the isomorphism $T$ of twisted $\mathbb{Z}_2$-graded Weil complexes given by \eqref{tdualitylinear}, in the obvious way.
\end{theorem}

\begin{proof} Under the identification of $\Gamma\bigg( \big((TM \oplus \mathbb{I}_{\theta^*}\oplus \mathbb{I}_{\hat{\theta}^*})  \oplus (T^*M \oplus \mathbb{I}_{\theta} \oplus \mathbb{I}_{\hat{\theta}})\big)^{\TT}_{\hat{\TT}-\text{bas}} \bigg)$ with a subspace of 
$$\Gamma\bigg(\big(T(E\TT \times M_{\hat{\TT}}) \oplus T(E\TT \times M_{\hat{\TT}})\big) / \TT\bigg),$$
and similarly for $\Gamma\bigg(\big((TM \oplus \mathbb{I}_{\theta^*}\oplus \mathbb{I}_{\hat{\theta}^*})  \oplus (T^*M \oplus \mathbb{I}_{\theta} \oplus \mathbb{I}_{\hat{\theta}})\big)^{\hat{\TT}}_{\TT-\text{bas}} \bigg)$, this follows from the Cavalcanti-Gualtieri isomorphism \eqref{cavgual}. \end{proof}

\section{Concluding remarks}
In this paper, we have significantly generalized the topological study of T-duality of $S^1$-manifolds \cite{BEM, BHM, BS, MW12}, to the case of arbitrary circle actions, where both the spacetime $X$ and the T-dual spacetime $\hat X$ can be singular spaces with nontrivial fluxes. One direction for future study is to replace circle actions with higher dimensional torus actions with nontrivial fixed-point sets. Another direction is to consider an analogue of the results of \cite{LM15} in the setting of non-free circle actions. In that paper, we showed that if $X$ and $\hat{X}$ are a T-dual pair of principal circle bundles over the same base with T-dual fluxes $H$ and $\hat{H}$, there is a T-duality isomorphism between appropriate vertex algebras coming from a twisted version of the chiral de Rham complexes on $X$ and $\hat{X}$. We expect that in the case of non-free circle actions, there will be a similar isomorphism of twisted equivariant chiral de Rham complexes, and twisted chiral equivariant cohomology. This is a $\mathbb{Z}_2$-graded version of the chiral equivariant cohomology which was introduced in \cite{LL07}, and the twist comes from replacing the differential $d$ by $d+H$, as in the Weil model of twisted equivariant cohomology.

\end{document}